\newif\ifislncs
\islncsfalse

\ifislncs 
\documentclass{../llncs/llncs}
\else
\documentclass[11pt]{article}
\usepackage[margin=1.5in]{geometry}
\fi

\usepackage{ifthen}

\newboolean{lncs}   
\ifislncs 
\setboolean{lncs}{true}
\else
\setboolean{lncs}{false}
\fi

\newboolean{short}   
\setboolean{short}{true}
\newboolean{linenumbers}   
\setboolean{linenumbers}{false}

\ifthenelse{\boolean{lncs}}{
  \let\doendproof\endproof
  \renewcommand\endproof{~\hfill$\qed$\doendproof}
}{
  \usepackage{amsmath,amssymb} 
  \usepackage{amsthm}
}

\usepackage[utf8]{inputenc}
\usepackage{graphicx}
\usepackage{subfig}
\usepackage{url}
\usepackage[colorlinks,pagebackref]{hyperref}

\usepackage{cite}
\usepackage{mathtools}
\usepackage[ruled,vlined,linesnumbered]{algorithm2e}
\usepackage{paralist}
\usepackage{todonotes}
\usepackage{xifthen}
\usepackage{thmtools, thm-restate}
\usepackage{xcolor}

\usepackage[capitalize]{cleveref}
\crefname{theorem}{Theorem}{theorems}
\crefname{definition}{Definition}{definitions}
\crefname{proposition}{Proposition}{propositions}
\crefname{lemma}{Lemma}{lemmas}
\crefname{corollary}{Corollary}{corollaries}
\crefname{claim}{Claim}{claims}
\crefname{observation}{Observation}{observations}
\crefname{fact}{Fact}{facts}
\crefname{dfn}{Definition}{definitions}
\crefname{obs}{Observation}{observations}
\crefname{pb}{Problem}{problems}

\newcounter{maincounter} 
\ifthenelse{\boolean{lncs}}{
}{

\newtheorem{lemma}[maincounter]{Lemma} 
\newtheorem{corollary}[maincounter]{Corollary}

}
\newtheorem{observation}[maincounter]{Observation}

\ifthenelse{\boolean{linenumbers}}{ 
\usepackage[left, pagewise]{lineno}
\linenumbers
\newcommand*\patchAmsMathEnvironmentForLineno[1]{%
  \expandafter\let\csname old#1\expandafter\endcsname\csname #1\endcsname
  \expandafter\let\csname oldend#1\expandafter\endcsname\csname end#1\endcsname
  \renewenvironment{#1}%
     {\linenomath\csname old#1\endcsname}%
     {\csname oldend#1\endcsname\endlinenomath}}%
\newcommand*\patchBothAmsMathEnvironmentsForLineno[1]{%
  \patchAmsMathEnvironmentForLineno{#1}%
  \patchAmsMathEnvironmentForLineno{#1*}}%
\AtBeginDocument{%
\patchBothAmsMathEnvironmentsForLineno{equation}%
\patchBothAmsMathEnvironmentsForLineno{align}%
\patchBothAmsMathEnvironmentsForLineno{flalign}%
\patchBothAmsMathEnvironmentsForLineno{alignat}%
\patchBothAmsMathEnvironmentsForLineno{gather}%
\patchBothAmsMathEnvironmentsForLineno{multline}%
}
}{
}

\makeatletter
\newdimen\commentwd
\let\oldtcp\tcp
\def\alignedtcp*[#1]#2{
\setbox0\hbox{#2}%
\ifdim\wd\z@>\commentwd\global\commentwd\wd\z@\fi
\oldtcp*[r]{\leavevmode\hbox to \commentwd{\box0\hfill}}}

\let\oldalgorithm\algorithm
\def\algorithm{\oldalgorithm
\global\commentwd\z@
\expandafter\ifx\csname commentwd@\romannumeral\csname c@\algocf@float\endcsname\endcsname\relax\else
\global\commentwd\csname commentwd@\romannumeral\csname c@\algocf@float\endcsname\endcsname
\fi
}
\let\oldendalgorithm\endalgorithm
\def\endalgorithm{\oldendalgorithm
\immediate\write\@auxout{\gdef\expandafter\string\csname commentwd@\romannumeral\csname c@\algocf@float\endcsname\endcsname{%
\the\commentwd}}}

\renewcommand{\O}{\ensuremath{\mathcal{O}}}

\newcommand{\E}{\ensuremath{\mathcal{E}}}

\newcounter{algo} 
\newcounter{subalgo}[algo] 
\newcommand{\G}{\ensuremath{\mathcal{G}}}
\newcommand{\V}{\ensuremath{\mathcal{V}}}
\renewcommand{\H}{\ensuremath{\mathcal{H}}}
\newcommand{\C}{\ensuremath{\mathcal{C}}}
\newcommand{\N}{\ensuremath{\mathcal{N}}}

\title{On the Distributed Computation of Fractional \\ Connected
  Dominating Set Packings}

\newcommand{\allauthors}{Fabian Fuchs,
 Matthias Wolf}
\newcommand{\allinstitute}{{Karlsruhe Institute for Technology } \\
  {\normalsize   Karlsruhe, Germany }}
\newcommand{\allemails}{{\normalsize fabian.fuchs@kit.edu, matthias.wolf3@student.kit.edu } }

\ifthenelse{\boolean{lncs}}{
\author{
 \allauthors
}
\institute{
  \allinstitute \\
    \email{\allemails}   
}
}{
\author{
 \allauthors \\
 \allinstitute \\
 \allemails 
}
}

\date{\today}

\begin{document}

\maketitle
\pagestyle{plain}

\begin{abstract}
One of the most fundamental problems in wireless networks is to
achieve high throughput. Fractional Connected Dominating Set (FCDS)
Packings can achieve a throughput of~$\Theta(k/\log n)$ messages for
networks with node connectivity~$k$, which is optimal regarding
routing-based message transmission. FCDS were proposed by
Censor-Hillel \emph{et~al.} [SODA'14,PODC'14] and are a natural
generalization to Connected Dominating Sets (CDS), allowing each node
to participate with a fraction of its weight in multiple FCDS. Thus,
$\Omega(k)$ co-existing transmission backbones are established, taking
full advantage of the networks connectivity.  We propose a modified
distributed algorithm that improves upon previous algorithms for
$k\Delta \in o(\min\{\frac{n \log n}{k} ,D,\sqrt{n \log n} \log^*
n\}\log n)$, where $\Delta$ is the maximum node degree, $D$ the
diameter and $n$ the number of nodes in the network. We achieve this
by explicitly computing connections between tentative dominating sets.
\end{abstract}


\section{Introduction}
\label{sec:introduction}

Wireless ad hoc and sensor networks are used to monitor the
environment, industry processes and even large parts of
infrastructure. In order to cope with the growing networks size and
its demand for efficient communication throughout the network,
algorithms and protocols that utilize the capacity available in the
network optimally are required. One of the standard methods to manage
high throughput in the network is to compute a backbone
structure. Recently, Censor-Hillel \emph{et al.}
\cite{chgk-dcd-14,chgk-anpvc-14} proposed an algorithm that allows to
build a network topology based on Fractional Connected Dominating Set
(FCDS, see \cref{sec:preliminaries} for a definition), which can be
seen as a generalized Connected Dominating Set (CDS).  Such fractional
connected dominating sets can be used to achieve a broadcast
throughput of $\Theta(k/\log n)$ messages in networks with $n$ nodes
and vertex-connectivity $k$, which is optimal regarding routing-based
approaches \cite{chgk-anpvc-14}. This improves on the standard method
of using one backbone by, intuitively, replacing it with as many
fractional backbones as the network can fit due to its
connectivity. To give further intuition, we show an example network
that admits multiple FCDS (a so-called FCDS packing) in
\cref{fig:fcds-illustration}.
In this work we propose an improved version of the distributed
algorithm originally proposed in \cite{chgk-dcd-14}. Our algorithm is
especially beneficial for future large-scale wireless networks, as
such networks are expected to consist of a huge number of small
wireless nodes deployed on a relatively large area. We present a
distributed algorithm that computes a FCDS packing by explicitly
computing the connector paths between not-yet connected components of
the respective FCDS. Our algorithm runs in the message-passing model
V-$\mathcal{CONGEST}$ and has a round complexity of
$\O(\log^2 n (\min\{\frac{n \log n}{k},D + \sqrt{n \log n} \log^* n\}
+ k\Delta))$.
This improves the runtime of previously
$\O(\log^3 n \cdot \min\{\frac{n \log n}{k},D + \sqrt{n \log n} \log^*
n\})$
for large and relatively sparse networks with moderate connectivity
$k$. More precisely our variant leads to an improvement regarding the
runtime if
$k\Delta \in o(\min\{\frac{n \log n}{k} ,D,\sqrt{n \log n} \log^*
n\}\log n)$,
which is for example true for logarithmic $k$ and $\Delta$ combined
with a diameter such that $k\Delta \in \O(D)$ (e.g., in the order of
$\sqrt{n}$). An example of a network with network parameters
beneficial for our algorithm is described in
\cref{sec:example-network}. Intuitively the achieved complexity is
beneficial for networks with large diameter and moderate density, and
generally for large but sparse networks.

Our algorithm is based on the virtual graph structure Censor-Hillel
\emph{et al.} \cite{chgk-anpvc-14,chgk-dcd-14} use to compute a FCDS
packing. In their distributed implementation they do not explicitly
compute the connector paths between the components but rely on the
fact that a sufficient number of paths exist, which requires
additional coordination within tentatively established dominating sets
(so-called components). In \cite{chgk-dcd-14}, the approach of
explicitly computing the paths is rendered as probably too
expensive. However, our algorithm improves the runtime while
explicitly computing possible connector paths between tentative
components. Let us now introduce the layered approach used in both
their algorithm and ours. For each node in the network we
introduce a set of $\O(\log n)$ virtual nodes, each virtual node shall
be assigned to one FCDS, resulting in a FCDS packing in which each
FCDS has at least weight~$1/\O(\log n)$. The virtual nodes are
assembled to a virtual graph~$\G$, and partitioned in layers such that
there are 1 to 3 copies of a node in each layer. A nodes copy is
connected to all copies of the nodes neighbors in $\G$
(cf. \cref{sec:preliminaries} for details). Now, layer for layer, the
virtual nodes are assigned different classes, each of which shall
result in a FCDS once the algorithm finishes. Using the first half of
the layers, dominating sets are formed, which are then connected by
selecting so-called connector paths between components using nodes
from the remaining layers. 

The improvement in our algorithm is achieved by improving the process
of how connector paths are matched to existing components. In
\cite{chgk-dcd-14}, paths are matched to components by building a
so-called bridging graph. In the bridging graph, whole components
(which may span large parts of the network) are simulated by a virtual
node and participate in a matching. Thus, the matching algorithm must
communicate through the component, which may require time in the order
of $\Omega(D)$ in each step of the matching.  In our algorithm we
construct a matching graph that can directly be executed by
transmitting one message to each neighbor in each round. 
%
%
Overall finding the connector paths
requires $\O(k\Delta \log n)$ rounds for each one of the $\O(\log n)$
layers. Once the connector paths are found, they can be used to
connect the components, in order to create multiple CDS is the virtual
graph. By translating each CDS to a FCDS in the network, $\Theta(k)$
communication backbones are established in the network, each with
weight $1/\log n$.

\noindent
\textbf{Related Work:}
\label{sec:related-work}
Research on FCDS was started by Censor-Hillel \emph{et al.}
in~\cite{chgk-anpvc-14}. They propose a centralized method to
compute CDS partitions of size $\Omega(k/\log^5 n)$, as well as FCDS
packings of size $\Omega(k/\log n)$, where $k$ is the
vertex-connectivity of the network.  FCDS packings are the natural
generalization of CDS partitions, which allows each node to
participate in multiple CDS with a fractional weight between 0 and 1
per CDS such that the sum of the weights is at most 1. Their approach
is based on a layered virtual graph, consisting of~$\Theta(\log n)$
virtual copies of each node. Each virtual node selects one
of~$\Theta(k)$ classes, which form the FCDSs later. Using the
first~$\log n$ layers they achieve domination by assigning random
classes to the nodes, the remaining layer are used to connect the
existing dominating sets to one connected dominating set per
class. Additional to computing FCDS packings, they show that the
broadcast throughput using a FCDS packing of size $\Omega(k/\log n)$
is $\Omega(k/\log n)$ messages per round, which is optimal if
restricted to routing-based approaches. In contrast to, for example,
network coding \cite{fs-ncf-07}, such approaches consider messages as
atomic tokens and use simple store-and-forward methods to route the
message.  The throughput achievable though routing-based approaches is
a logarithmic factor less than the $\Theta(k)$ messages that can be
achieved using network coding \cite{chgk-anpvc-14}, however, without
introducing further challenges (see \cite{fs-ncf-07} for example).
Ene, Korula and Vakilian \cite{ekv-iaacd-13} consider FCDS packings
under the constraint that each node has a capacity. Also using
centralized algorithms they compute FCDS packings of size $\Omega(k)$
for planar and minor-closed families of graphs, and $\Omega(k/\log n)$
for the general case.  The first distributed implementation is again
due to Censor-Hillel, Ghaffari and Kuhn \cite{chgk-dcd-14}. In this
work they consider both vertex- and edge-connectivity. For
vertex-connectivity they compute a FCDS packing (or a fractionally
disjoint weighted dominating tree, which is a similar concept) of size
$\Omega(k/\log n)$, building on the initial approach in
\cite{chgk-anpvc-14} as we do in this work.

\ifthenelse{\boolean{short}}{}{
Other related topics are connected dominating sets
\cite{gk-aacds-98,db-rahnmcds-97}, as well as dominating set
partitions (see e.g. \cite{fhks-adn-02}).  Regarding the application
of FCDS to increase the throughput, related works are multi-message
broadcasting in wireless networks \cite{bii-mcmrn-93} and network
coding \cite{fs-ncf-07}.
}


\section{Preliminaries}
\label{sec:preliminaries}

Our algorithms operate on the communication graph $G=(V,E)$ of a
wireless sensor network, where $V$ is the set of nodes or actors in
the network, and $E$ the set of edges. An edge $e=(u,v) \in E$ is in
the communication graph if $v,u \in V$ can communicate in the
network. We assume the communication to be bidirectional, and hence
the communication graph to be undirected.
We assume a standard message-passing model, known as
(V-)$\mathcal{CONGEST}$. Communication is based on synchronous rounds, in which
each node can receive the messages of its neighbors as well as
transmit one identical message to all neighbors itself. In contrast to
the E-$\mathcal{CONGEST}$ model, in which the nodes may transmit different
messages to their neighbors, the congestion is on the nodes instead of
the edges of the corresponding communication graph. Note that this
fits the broadcast nature of wireless networks.

A \emph{dominating set} $S \subseteq V$ is a set of nodes in the
network such for each node $v \in V$ it holds that either $v \in S$ or
a neighbor of $v$ is in $S$. If such a set is connected we denote it
by connected dominating set (CDS). \emph{Fractional CDS} (FCDS)
packings are the natural generalization of CDS. In an FCDS packing,
each node can participate in multiple FCDSs with a weight
$x_i \in [0,1]$ for each FCDS, such that $\sum_i x_i \leq 1$ for each
node. 
\begin{figure}[bth]
  \centering
  \includegraphics[width=0.7\linewidth]{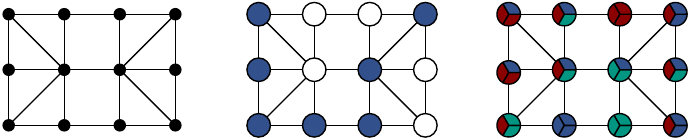}
  \caption{\small From left to right: A 3-vertex connected graph $G$, a (not
    optimal) connected dominating set in $G$ marked by blue nodes, and
    a FCDS packing in G establishing 3 fractional connected dominating
    sets.}
  \label{fig:fcds-illustration}
\end{figure}
The virtual graph $\G = (\V, \E)$ used in the construction of the FCDS
was first introduced in \cite{chgk-anpvc-14}. For each node $v \in V$
we introduce $\O(\log n)$ copies in the set of virtual nodes $\V$.
Each copie of $v$ is connected to all other copies of $v$ in $\V$ and
to each copy of a neighbor $w \in V$ of $v$ in $\G$. We denote the
neighbors of $v$ in $G$ by $N_v$, and in $\G$ by $\N_v$.  In contrast
to the original description, which was also used for the first
distributed implementation in \cite{chgk-dcd-14}, we use $3L$ copies
of each node instead of $4L$ for $L\in\O(\log n)$, however, this is a
minor technical detail.  In total the virtual graph has $3L$ copies of
the original graph, plus some additional edges. We subdivide the
virtual graph in layers and call the first $L$ copies of $V$ in $\V$
the \emph{lower layers}. Each so-called \emph{upper layer} consists of
two copies of $V$. We call the nodes of the first copy \emph{type-1}
nodes, and the nodes of the second copy \emph{type-2} nodes. For each
layer $l$ we denote the nodes of layers $1$ to $l$ by $\V_l$ and the
subgraph induced by these nodes by $\G_l$.

As we compute multiple FCDS simultaneously in the virtual graph, we
distinguish each FCDS by a class $i \leq t \in \Theta(k)$.  We denote the
subset of nodes of class $i$ up to layer $l$ by $\V^i_l$, and the
induced subgraph by $\G^i_l$.  We use $\Psi(v_l) = v$ to project from
nodes (or a set of nodes) of the virtual graph to the corresponding
real node(s). Throughout the rest of the paper we shall use the term
node to refer to virtual nodes in $\G$, and real node to refer to a
node of network.
%
%
During the execution of the algorithm, we aim at connecting
not-yet-connected components of the dominating set of a class $i$ to
other components of the same class. Given a connected component $\C$
we use so-called connector paths to identify vertices that could
connect $\C$ to another component $\C'$ with
$\Psi(\C) \cap \Psi(\C') = \emptyset$, both of the same class. In
compliance with \cite{chgk-anpvc-14} we call a path $P$ a
\emph{connector path} for component $\C$ if it satisfies the following
conditions:
\begin{enumerate}[a)]
\item \label{cond:cp-1} $P$ has one endpoint in $\C$ and the other
  endpoint in $\C'$
\item \label{cond:cp-2} $P$ has at most two internal vertices and
\item \label{cond:cp-3} P cannot be shortened, i.e., for
  $P= (s, v, w, u)$ with $s \in \C$ and $u\in\C'$, $u$ does not have a
  neighbor in $\C'$ and $v$ does not have a neighbor in $\C$.
\item \label{cond:cp-4} if $P= (s, v, w, u)$ with $s \in \C$ and
  $u\in\C'$ has two internal vertices, $v$ is of type-2 and $w$ of
  type-1.
\item \label{cond:cp-5} if $P = (s, v, u)$ with $s \in \C$ and $u\in\C'$
  has one internal vertices, $v$ is of type-1.
\end{enumerate}
Connector paths can have length two or three as the components of
each class are already dominating. We call connector paths of length
two \emph{short} and those of length three \emph{long}. 
For a path $(v_1, v_2,\dots, v_{i-1},v_i)$, we call the set of nodes
$\{v_2,\dots,v_{i-1}\}$ the \emph{internal} nodes. We call a set of paths
$\{P_1,\dots,P_j\}$ \emph{internally vertex-disjoint} if the internal
nodes of $\{P_1,\dots,P_j\}$ are mutually disjoint.

The following lemma states that we always find at least $k$ connector
paths for each component in a $k$-connected graph.
\begin{lemma}[Lemma 4.3 in \cite{chgk-dcd-14}]
  \label{lem:k-connector-paths}
  For each component $\C$ of an arbitrary class $i$ at an arbitrary
  level $l$ it holds that $\C$ has at least $k$ internally vertex
  disjoint connector paths
\end{lemma}
Note that the algorithms proof of correctness requires the connector
paths of one component to be internally vertex-disjoint. We assume our
connector paths to have this property in the following section. It is
easy to verify that enough short connector paths are available. For
long connector paths, we shall explain how a sufficient number of
internally vertex-disjoint long connector paths can be found in
\cref{sec:find-long-conn}. In the virtual graph $\G$, our algorithm
computes a CDS partition, as each node may select exactly one class
$i$. Let us now briefly consider how this translates to a FCDS packing
in $G$.
%
%
Let $v \in G$ be a node of the network and $v_1,\dots,v_{3L}$ the
corresponding virtual nodes in $\G$. Given $\Theta(k)$ CDSs in $\G$, we can construct a FCDS in $G$ by weighting
the class of each virtual node $v_i$ by $1/3L$ at the real node
$v$. As there are $3L$ virtual copies of $v$, the weight constraint is
satisfied, and the CDS partition translates to a FCDS packing.



\section{Distributed FCDS Computation}
\label{sec:distr-fcds-comp}

Our algorithm consists of two main components to construct
$t = \Theta(k)$ CDSs in $\G$. Recall that we
assign each virtual node to one of $t$ classes, which shall form the
CDSs after the execution of the algorithm.  The first $\O(\log n)$
layers of virtual nodes establishes that each class dominates the
whole graph with high probability (cf. \cref{lem:dominance}).  This is
surprisingly simple and can be achieved by having each virtual node
select one of the classes at random.  For the second $\O(\log n)$ layers
we aim at connecting a constant fraction of the connected components
in each layer (with constant probability). This leads to connectivity
of each class with high probability after the last layer, yielding the
desired CDSs.
The existing distributed algorithm to compute FCDS packings uses the
same approach for the lower layers, and (essentially) matches existing
components in each of the upper layers without computing the connector
paths.  Our approach on the other hand explicitly computes the (long)
connector paths by constructing a helper graph in which a matching
algorithm finds $\Omega(k)$ such paths. Thus we do not require
communication through existing components, which is beneficial for
many networks, especially if they are large with respect to the diameter.
Our algorithm consists of the following steps. Note that the overall
design of the algorithm is similar to that of \cite{chgk-dcd-14},
however, we use a different method to connect the components of each
class, which is one of the key parts of the algorithm.

\renewcommand{\paragraph}[1]{\hspace*{-4mm}
\textbf{#1)}}
\ifislncs 
\renewcommand{\subparagraph}[1]{\hspace*{0mm}\textbf{#1)}}
\else 
\renewcommand{\subparagraph}[1]{\hspace*{0mm}\textbf{#1)}}
\fi
\newcounter{algostepscounter}
\newcommand*{\algolabel}[1]{B.\refstepcounter{algostepscounter}\thealgostepscounter \label{#1}}
\newcommand*{\algoref}[1]{\hyperref[#1]{B.\ref{#1}}}

\paragraph{A} 
\phantomsection\label{algo:step-a}
Each virtual node in the lower layers randomly selects one of the $t$
classes. This leads to domination of each class whp, cf.
\cref{lem:dominance}.

\paragraph{B} 
\phantomsection \label{algo:step-b}
For each upper layer $l$ from $L$ to $2L$ we try to connect existing
connected components of each class in the nodes of layers $1$ to $l-1$
using nodes of layer $l$.  We call the nodes of the previous layers
$1$ to $l-1$ \emph{old nodes} and the nodes of layer $l$ \emph{new
  nodes}. For each layer we execute steps \algoref{algo:identify-con-components}
to \algoref{algo:type-2-random}.

\subparagraph{\algolabel{algo:identify-con-components}} Identify
connected components of old nodes. We use the protocol described in
\cite{chgk-dcd-14}.
To be self-contained, we describe the protocol in
\cref{sec:ident-conn-comp}.

\subparagraph{\algolabel{algo:type-1-random}} Let nodes of type-1
select a random class

\subparagraph{\algolabel{algo:find-long-conpaths}} For each class $i$:
If the nodes component is not yet connected by short connector paths,
find $\Omega(k)$ internally vertex-disjoint long connector paths.  We
construct a helper graph $\H_i$ and run a simple matching algorithm to
find the long connector paths. For details on this step we refer to
Section~\ref{sec:find-long-conn}.

\subparagraph{\algolabel{algo:type-2-random}} If the type-2 node
is on long connector paths, the node discards the~paths for which the
type-1 node selected a wrong class, and selects the class of one of
the remaining paths at random
. If no path remains a
random class is selected.

After executing this algorithm each virtual node in $\G$ is assigned
itself to one of the $t$ classes.  Each class dominates the whole
graph (Step~\hyperref[algo:step-a]{A}) and is connected (Step~\hyperref[algo:step-b]{B}).  Thus, the nodes computed
$t = \O(k)$ CDSs in the virtual graph $\G$. The CDSs can be converted
to one FCDS of size $\Omega(k/\log n)$ by assigning each CDS a weight
of $1/3L$ (cf. \cref{sec:preliminaries}).
Note that the matching algorithm in Step~\algoref{algo:type-1-random}
matches type-1 with type-2 nodes, thus it does not require
communication and coordination within large components. Let us now
briefly reference the result that achieves dominance in the lower
layers. 
\begin{lemma}[\small Lemma~4.1 \cite{chgk-dcd-14}]
  \label{lem:dominance}
  For each class $i$, $\V^i_l$ is a dominating set in $\G$ w.h.p.
\end{lemma}

\ifthenelse{\boolean{lncs}}{}{
The proof idea is based on the fact that, for class $i$ and a
node $v \in V$, the probability that $v$'s virtual copy on layer $l$
selects $i$ is at least $1/t = 1/\O(k)$. As each node has at least $k$
neighbors on $l$, this yields constant probability per layer, and w.h.p.
over all $\log n$ layers.
}


\section{Finding Connector Paths}
\label{sec:find-long-conn}

In this section we show how our algorithm computes internally
vertex-disjoint connector paths for each component in order to connect
a constant fraction of the components in each upper layer. We begin
this section by giving a high-level proof showing that we can indeed
connect a constant fraction of the components with each new layer. In
the next sections we introduce the necessary tools and prove the
remaining results. In \cref{sec:virt-graph-hi} we introduce the graph
$\H_i$, which helps to reduce the problem of finding long paths for
each component to a matching problem. The matching problem is
discussed in more detail in \cref{sec:matching}.

As introduced in \cref{sec:preliminaries}, connector paths can have
one or two internal nodes, we call them short and long connector
paths, respectively.  To prove correctness for the algorithm, it must
hold for each one of the upper layers that at least a constant
fraction of the components (formed by old nodes) of each class are
connected to another component of the same class using nodes from the
current layer with at least constant probability. 
%
%
We shall now state the overall result of this section, which was first
obtained and proven in \cite{chgk-dcd-14}.  Note that there is a minor
flaw in the original proof regarding the number of missing connections
$M_l$ in layer $l$, see \cite[p. 21]{wolf-ba-14} for details and a
corrected proof. Due to space constraints we sketch the proof of the
following lemma in \cref{sec:ml-lemma}

\begin{lemma}[Lemma~4.4 in \cite{chgk-dcd-14}]
  \label{thm:connect-components-in-layer}
  Let $l \in [L,3L]$. Then $M_{l+1} \leq (1-\delta) M_l$ with
  probability at least $\rho$.
\end{lemma}


\subsection{Helper Graph $\mathcal{H}_i$}
\label{sec:virt-graph-hi}

Finding internally vertex-disjoint long connector paths is only
relevant if a component has less than $k/2$ internally vertex-disjoint
short connector paths. As each component has at least $k$ internally
vertex disjoint connector paths, the component must have at least
$k/2$ long connector paths in this case. We introduce a helper graph
in this section, which is defined such that a maximum matching in this
graph corresponds to finding a maximum number of internally
vertex-disjoint long connector paths.

For each class $i$ on an upper layer $l$ we define the helper graph
$\H_i^l$ as the~union of the helper graphs $\H_i^l[\C]$ constructed
for each component $\C$ of class $i$ on layer~$l$. Note that although
the helper graphs are constructed for each layer, we omit $l$ in the
following as the helper graphs are used only in the layer in which
they are constructed. Thus, we always refer to the helper graph of the
current level.

We define the helper graph $\H_i[\C]$ for class $i$ and
component $\C$.  For each \mbox{type-2} node $v$ of layer $l$ we add a node
$v_{\C}$ to $\H_i[\C]$ iff the following conditions are met
\begin{enumerate}[1)]
\item \label{cond:hi-1} $\Psi(v) \not \in \Psi(\C)$
\item \label{cond:hi-2} $v$ has a neighbor in $\C$
\item \label{cond:hi-3} $v$ does not have a neighbor belonging to
  another component of class $i$
\end{enumerate}
For each node $v_C$ we added to $\H_i[\C]$, we add for each type-1
neighbor $w$ of $v$ a node $w_C$ to $\H_i[\C]$, if $w$ has a neighbor
in another component $\C'$ of class $i$ but no neighbor in component
$\C$. Intuitively , this procedure ensures that we added the potential
long connector path of component $\C$ to $\H_i[\C]$ using a type-2
node as the node closer to $C$ and a type-1 node as the node closer to
the neighboring component of the same class.  An edge between $v_C$
and $w_C$ is added to $\H_i$ as there is an edge between $v$ and $w$
in $\G$. 
\begin{figure}[h]
  \centering
  \includegraphics[width=0.7\linewidth]{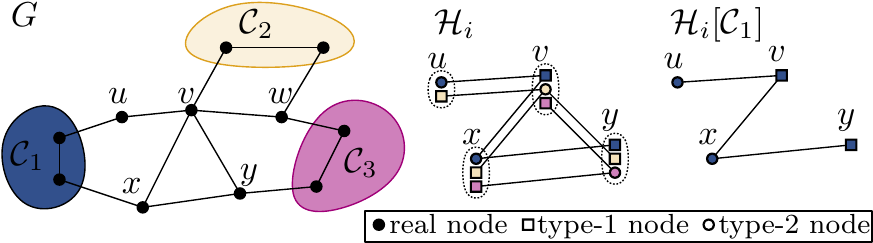}
  \caption{\small A graph G with three components of a class $i$, along with
    the helper graph $\H_i$ and $\H_i[\C_1]$ restricted to component
    $C_1$. Note that $w$ is not in the helper graph as it is on a long
  connector path.}
  \label{fig:hi-illustration}
\end{figure}
We illustrate the construction of the helper graph $\H_i$ in
\cref{fig:hi-illustration}.  The connection between long connector
paths of a component and edges in $\H_i$ is shown in the following
lemma.

\begin{lemma}
  \label{lem:hi-edge-equals-connector-paths}
  There is an edge $(v_C,w_C)$ in $\H_i[\C]$ iff there is a long
  connector path from $\C$ to another component of class $i$ through
  $v$ and $w$ on the current layer.
\end{lemma}

\begin{proof}
  Let us first assume an edge $(v_C,w_C)$ is added to $\H_i[\C]$. Then
  $v$ has at least one neighbor in $\C$, which we denote by $s$. Also,
  $v$ has a neighbor $w$ (of type-1) which does not have a neighbor in
  component $\C$ but has at least one in a component $\C' \not= \C$ of
  class $i$. Let us denote this neighbor by $u$. We claim that
  $P = s,v,w,u$ is a long connector path (cf. \cref{sec:preliminaries}
  for the definition), which holds as \ref{cond:cp-1}) P has one
  endpoint in $\C$, the other in $\C' \not= \C$ of class $i$,
  \ref{cond:cp-2}) P has two internal vertices, \ref{cond:cp-3}) P
  cannot be shortened as $v$ does not have a neighbor in a component
  $\C' \not= \C$ of class $i$ and $w$ does not have a neighbor in
  $\C$, and \ref{cond:cp-4}) $v$ is of type-2 while $w$ is of type-1.

  Let us now assume we have a long connector path $P = s,v,w,u$. It
  holds that \ref{cond:hi-1}) $w$ does not have a neighbor in $\C$,
  \ref{cond:hi-2}) $s$ is in $\C$, and \ref{cond:hi-3}) $w$ is not in
  $\C$ and $v$ of type-2. Thus, $v_C$ is added to
  $\H_i[\C]$. Also, $w$ is of type-1 and has a neighbor~$u$~that
  is in another component $\C' \not= \C$ of class $i$ but no neighbor
  of component $\C$, which implies that $w_C$ and the corresponding
  edge $(v_C,w_C)$ are added as well.
\end{proof}

The matching algorithm is executed on $\H_i$, which is the union of
helper graphs for each component, however, observe that we know for
each edge in $\H_i$ from which component it is induced.

\begin{observation}
  \label{obs:hi-components-disjoint}
  Given an arbitrary layer $l \geq L$, a class $i$, and the
  corresponding helper graph $\H_i$.  Then each edge in the helper
  graph can be attributed to exactly one component $\C$.
\end{observation}

We have shown that the construction ensures that there is a vertex
disjoint long connector path through $v$ and $w$ for component $\C$
iff there is an edge between $v_C$ and $w_C$ in $\H_i[\C]$. Thus a
matching induces long connector paths. We shall argue in
\cref{sec:matching} that we can compute a matching of size $\Omega(k)$
in $\H_i$ for each component of class $i$ with $\Omega(k)$ long
connector paths. However, let us first describe the distributed
algorithm to construct $\H_i$

\subsection{Distributed Construction $\H_i$}
\label{sec:constructing-hi}

Due to Step~\algoref{algo:identify-con-components} of the algorithm,
which is executed for each layer before constructing the graph $\H_i$,
each type-2 node $v$ knows the classes and components of its
neighbors. Thus $v$ can decide whether a node $v_C$ should be added
for neighboring components $\C$. Note that due to
\cref{lem:type-2-node-on-leq-one-path} a type-2 node lies only on one
long connector path for each class, however, up to $t$ 
components may have a long connector path through $v$, see
\cref{obs:type-2-node-leq-delta-components}. If $v$ adds $v_C$ for a
component $\C$ to $\H_i[\C]$, it transmits this information along with the
class $i$ and the id of $\C$ (which is also used in
Step~\algoref{algo:identify-con-components}) to its type-1
neighbors. These type-1 neighbors can now easily check whether they
have a neighbor in $\C$ and resign, or verify if at least one neighbor
is in another component of class $i$ due to information obtained
during Step~\algoref{algo:identify-con-components}. If so, $w$ adds
$w_C$ and the edge between $w_C$ and $v_C$~to~$\H_i[\C]$.

\begin{lemma}
  \label{lem:hi-construction-runtime}
  For each class $i$ on layer $l$ it requires $\O(\Delta)$ rounds
  to construct $\H_i$. 
\end{lemma}

\begin{proof}
  First note that each real node $v$ simulates exactly the two virtual
  copies of $v$ on layer $l$. Due to Conditions \ref{cond:hi-2}) and
  \ref{cond:hi-3}) in the definition of $\H_i$, the type-2 copy of a
  node participates in $\H_i$ only if its neighbors of class $i$
  belong to the same component $\C$. In this case, $v$ sends the id of
  $\C$, which requires one message.  After receiving these messages,
  each type-1 node $w$ transmits one message for each message they
  received from a type-2 node. Note that as $w$ has received at most
  one message from each neighbor, $w$ responds to at most $\Delta$
  messages. Hence, this results in $\O(\Delta)$ messages.
\end{proof}

The following observation follows from the fact that a type-2 node is
only added to $\H_i$ if all its neighbors of the class $i$ are in one
component, while one type-1 copy is added for each message received by
another type-2 node. It helps bounding the runtime of our matching
algorithm operating on $\H_i$.

\begin{observation}
\label{obs:hi-node-copies}
  For each node $v \in V$, there is at most one type-2 copy in $\H_i$, but up to
  $\Delta$ type-1 copies in $\H_i$.
\end{observation}

\subsection{Matching internal vertices}
\label{sec:matching}

Let us now consider how to distributively compute a matching of
cardinality $\Omega(k)$ in the helper graph $\H_i$ for each
component. We shall use this in the next section to prove that each
component finds a long connector paths with constant probability in
each layer.
We use a randomized distributed maximal matching algorithm, which was
proposed by Censor-Hillel \emph{et al.} in \cite{chgk-dcd-14} and is
based on Luby's distributed maximal independent set algorithm
\cite{l-spamis-86}. However, in our case each node in the helper graph
$\H_i$ is simulated by only one node at not by several nodes that have
to coordinate their actions through components. 
The algorithm makes use of the special structure of $\H_i$.

\begin{lemma}
\label{obs:helper-graph-bipartite}
The helper graph $\H_i$ is bipartite.
\end{lemma}

\begin{proof}
As described in the previous section, all nodes in the helper graph $\H_i$ are 
either added by a type-1 or type-2 node of $\G$. Hence, we may say that the 
nodes of the helper graph also have types, which induce a partitioning of the 
nodes of $\H_i$.
For two nodes $v_C$ and $w_C$ to be connected in $\H_i$, it must hold that $v$ 
is of type 2 and $w$ of type 1. Thus, all edges in $\H_i$ connect nodes of 
different types, which shows that the graph is bipartite. 
\end{proof}

Using this lemma the matching algorithm operates as follows.  A node
is \emph{active} exactly if none of the adjacent edges is matched, and
an edge is \emph{active} if both adjacent nodes are active. In each
round, we assign random numbers from a sufficiently large range to all
active edges such that no two edges have the same number whp. Since
$\H_i$ is bipartite, assigning the numbers is particularly easy as
each type-2 node can pick a number for each incident edge.  Each
active type-2 node then selects the edge with the largest number and
sends its choice to its neighbors. In this round only the selected
edges may be added to the matching. At this point, there is at most
one edge selected at each type-2 node. However, each type-1 node may
have received more than one proposal. To satisfy the matching
condition, each type-1 node that has received at least one proposal
picks the proposed edge with the largest number and adds it to the
matching. The two matched nodes and their edges become inactive. It
can be shown that after $\O(\log n)$ rounds all edges are deactivated
with high probability and a maximal matching is achieved. Let us
now show that such a maximal matching is of cardinality $\Omega(k)$ if
the corresponding component has $\Omega(k)$ long connector paths.

\begin{lemma}
  \label{lem:max-matching-const-approx}
  Given a component $\C$ of class $i$ with $\Omega(k)$ long connector
  paths. A maximal matching in $\H_i[C]$ is of cardinality at least
  $\Omega(k)$.
\end{lemma}

\begin{proof}
  If follows from \cref{lem:k-connector-paths} and the one-to-one
  correspondence of the long connector paths and the edges in $\H_i$
  of \cref{lem:hi-edge-equals-connector-paths} that there are
  $\Omega(k)$ independent edges in $\H_i[C]$. Thus, the maximum
  matching is of size at least $\Omega(k)$, as well as the maximal
  matching as it is a 2-approximation of the maximum matching.
\end{proof}

After showing that the matching is of sufficient size, we prove that
this allows us to identify the $\Omega(k)$ long connector paths for
each component.

\begin{lemma}
  \label{lem:max-matching-k-paths-per-component}
  Consider a component $\C$ of class $i$ with $\Omega(k)$ long
  connector paths. Then a maximal matching in $\H_i$ identifies
  $\Omega(k)$ long connector paths for $\C$.
\end{lemma}

\begin{proof}
  Let us consider a maximal matching in $\H_i$, and component $\C$ as
  required. According to \cref{obs:hi-components-disjoint}, we can
  consider the subgraph $\H_i[\C]$ of $\H_i$ corresponding to
  component $\C$ as disjoint from other parts of $\H_i$. Thus, the
  matching is maximal also in $\H_i[\C]$.  It holds by
  \cref{lem:max-matching-const-approx} that the size of the maximal
  matching is $\Omega(k)$.  It remains to show that two independent edges
  in $\H_i[\C]$ correspond to two internally vertex disjoint connector
  paths. Consider the edges $(v_C,w_C)$ and $(v'_C,w'_C)$ and assume
  the corresponding long connector paths with internal vertices $v,w$
  and $v',w'$ are not internally vertex disjoint. Thus, either $v=v'$
  or $w=w'$ which implies either $v_C=v'_C$ or $w_C=w'_C$. This
  contradicts the assumption as the edges are not independent. As each
  matched edge is independent, the set of matched edges in $\H_i[\C]$
  corresponds to a set of internally vertex-disjoint long connector
  paths of cardinality $\Omega(k)$.
\end{proof}

The correspondence between a maximal matching in $\H_i$ and long
connector paths in $G$ is depicted in \cref{fig:hi-matching}.
Let us now consider the number of time slots required to compute the
maximal matching. The matching algorithm is executed once for every
class $i$, and operates on the virtual graph $\H_i$. The next lemma
proves that $\O(\Delta)$ time slots are sufficient for each round of
the matching algorithm.
\begin{figure}[h]
  \centering
  \includegraphics[width=0.7\linewidth]{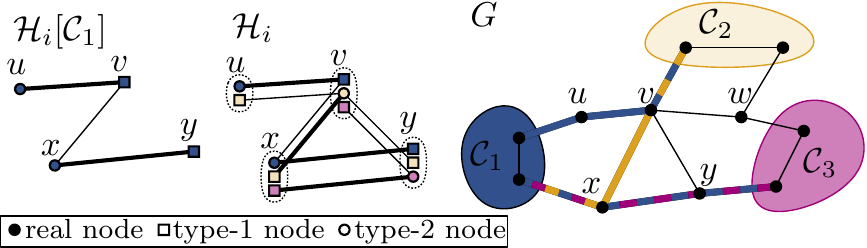}
  \caption{\small A maximal matching in the helper graph $\H_i$ is marked by
    bold lines. The corresponding long connector paths are marked in
    $G$.}
  \label{fig:hi-matching}
\end{figure}
\begin{lemma}
  \label{lem:matching-edge-used-leq-twice}
  In each round of the matching on $\H_i$ we transmit over
  each real edge at most twice in each direction, resulting in
  $\O(\Delta)$ time slots for each round.
\end{lemma}

\begin{proof}
  Let $v$ be an arbitrary real node, and note that there may be up to
  $\Delta$ copies of $v$ as type-1 node in $\H_i$, but only one copy
  of $v$ as type-2 node in $\H_i$, cf. \cref{obs:hi-node-copies}.
  Consider any real edge from $v$ to an arbitrary neighbor $w$. We may
  assume that there is at least one copy of the edge $(v,w)$ in the
  helper graph $\H_i$, as otherwise this edge is not used for the
  matching algorithm at all. Since the type-2 copy of $v$ sends only
  one message over one of its incident edges, it uses the edge $(v,w)$
  at most once.  After the type-1 copies of $v$ have received the
  messages from the type-2 nodes, they respond to one of them. Hence,
  each type-1 copy sends at most one message. It remains to show that
  no two type-1 copies use the same real edge. Assume that there were
  two type-1 copies of $v$ that transmit over the real edge
  $(v,w)$. This would imply that both type-1 copies have received a
  message from the type-2 copy of $w$ over the real edge
  $(w,v)$. However, we have shown above that each real edge is used at
  most once by the type-2 nodes, which contradicts our
  assumption. Thus, in one round each real edge transmits at most one
  message from a type-1 and one from a type-2 node, resulting in two
  messages per edge. As we operate in the V-$\mathcal{CONGEST}$ model,
  two messages per edge results in $\O(\Delta)$ time per round.
\end{proof}

It follows from \cite{l-spamis-86} that $\O(\log n)$ rounds are
sufficient to compute a maximal matching with high probability.

\begin{corollary}
  \label{cor:matching-runtime}
  Our distributed randomized matching algorithm computes a maximal
  matching in $\H_i$ in~$\O(\Delta\log n)$ time.
\end{corollary}

This implies that we can find $\Omega(k)$ long connector paths for all
components of one class that have less than $k/2$ short connector
paths in time $\O(\Delta \log n)$. As we have $t = \Theta(k)$ classes,
this results in $\O(k\Delta\log n)$ for
Step~B.\ref{algo:find-long-conpaths} on each layer.  Let us now prove
that the long connector paths can indeed be used to connect the
components with at least constant probability on each layer.

\subsection{From long connector paths to connected components}
\label{sec:from-connector-paths}

As components with at least $k/2$ short connector paths are connected
using those connector paths, we keep focusing on components with at
least $k/2$ long connector paths. In the previous section we showed
how to find $\Omega(k)$ vertex-disjoint long connector paths for each
such component. As type-1 nodes already selected a random class to
connect those components that have a sufficient number of short
connector paths, the class of the type-2 nodes on the current layer
remains to be selected.  Since each type-2 node lies on at most one
long connector path per class, there are at most $t$ long connector
paths per type-2 node. On these paths, however, the internal type-1
nodes may have chosen classes that differ from the class of the path.
Intuitively, this means that the path cannot be used to connect two
components of the same class since one of the internal nodes has
already picked the wrong class. Therefore, as described in
Step~\algoref{algo:type-2-random}, the type-2 nodes discard these long
connector paths and select the class of one of the remaining paths at
random. If no long connector paths remains, the node selects a random
class.

We show in this section that this is sufficient to guarantee that a
constant fraction of the components are connected with constant
probability. Let us consider an arbitrary component $\C$ of class
$i$. There are two challenges. The first is to show that each
connector path connects to another component of the same class with
probability in the order of $1/k$. This is non-trivial, as the type-1
node on each long connector path already selected a random class,
which upper bounds the probability by $1/t$. The second challenge is,
that the events that two type-2 nodes on different connector paths of
$\C$ selecting class $i$ are not necessarily independent. This can be
circumvented by using a tail bound, once the probability for each
event is upper and lower bounded independently of the outcome of other
events.  Let us now state the result.  The proof is based on
\cite{chgk-dcd-14} with some modifications. We sketch the main ideas
in \cref{sec:lcp-good}.
\begin{lemma}
  \label{prop:long-paths-constant-fraction}
  Given a component $\C$ of class $i$ on an upper layer $l$ with
  $\Omega(k)$ long connector paths. The probability that one of the
  long connector paths is good is at least $\delta$.
\end{lemma}


\section{Conclusion}
\label{sec:conclusion}

The algorithm presented in this work computes a fractional connected
dominating set packing in the $\mathcal{CONGEST}$ model of distributed
computation. It is based on an algorithm by Censor-Hillel, Ghaffari
and Kuhn \cite{chgk-dcd-14,chgk-anpvc-14}, however, our distributed
implementation computes the long connector paths explicitly, instead
of matching components under the assumption that sufficient long
connector paths exist. The runtime of our algorithm is
$\O(\log^2 n (\min\{\frac{n \log n}{k},D + \sqrt{n \log n} \log^* n\}
+ k\Delta))$
, which is beneficial for large networks with moderate density,
particularly if
$k\Delta \in o(\min\{\frac{n \log n}{k} ,D,\sqrt{n \log n} \log^*
n\}\log n)$.
We expect future large-scale wireless sensor networks to satisfy such
conditions.

\vspace{2mm}
\noindent\textbf{Acknowledgements: }
\label{sec:acknowledgements}
Parts of this work are based on the seconds authors BSc thesis
\cite{wolf-ba-14}.  We thank the German Research Foundation (DFG),
which supported this work within the Research Training Group GRK 1194
"Self-organizing Sensor-Actuator Networks".


{\small
\bibliography{abbrv,bib}
\bibliographystyle{splncs03}
}

\begin{appendix}
\section{Appendix}
\label{sec:appendix}

\subsection{Example network}
\label{sec:example-network}

Let us consider the following network. The nodes of the network are
distributed evenly on $d$ rings such that each node on a ring is
connected to its neighbors on the ring. Also, $d$ nodes, one from each
of the rings, for am clique. For $d=4$, an example of such a network
is displayed in \cref{fig:example-network}. 
\begin{figure}[h]
  \centering
  \includegraphics[width=0.5\linewidth]{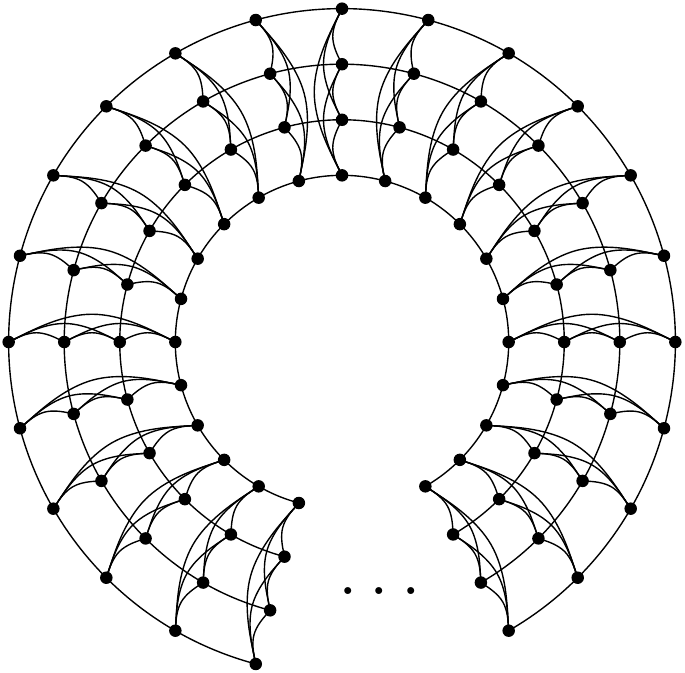}
  \caption{A 5-vertex connected network, with diameter in $\Omega(n)$.}
  \label{fig:example-network}
\end{figure}
We can see that such a network is $d+1$ connected, and has maximum
degree $d+1$. The diameter of the network, however, is in
$\Omega(n/d)$. Let us, for example, consider a network with
$d = \sqrt[3]{n}$. Such a configuration results in
$k \Delta = D = \O(n^{2/3})$.  Thus, the runtime of the proposed
modification yields an improvement over \cite{chgk-dcd-14}.

\subsection{Number of Components decreases}
\label{sec:ml-lemma}

In the following we sketch the proof of
\cref{thm:connect-components-in-layer}. This proves that the number of
components decreases in each iteration by a constant factor with at
least constant probability.

\begin{proof}[Sketch of proof, based on \cite{chgk-dcd-14,wolf-ba-14}]
  To proof the theorem we consider each component and show that the
  component is connected to another component by layer $l+1$ with
  constant probability.  This implies that a constant fraction of the
  existing components are connected by layer $l+1$ with constant
  probability.  Given a component $\C$ of class $i$ on layer $l$ and
  assume class $i$ has at least two components. It holds that $\C$ has
  at least $k$ connector paths, connecting $\C$ to another component
  of class $i$, according to \cref{lem:k-connector-paths}. There are
  two cases: Either at least $k/2$ of the paths are short connector
  paths, or at least $k/2$ of the paths are long connector paths.

  Let us first consider the case of at least $k/2$ short connector
  paths. This is intuitively the easier case, as only one node
  separates $\C$ from another component on $\Omega(k)$ paths.  Recall
  that each layer has two copies of each real node: a type-1 and a
  type-2 node. Let us consider only the type-1 node for now. According
  to \cref{lem:short-paths-constant-fraction} it is sufficient that
  all type-1 nodes select a random class to connect $\C$ to a
  neighboring component with constant probability in this case.
  Intuitively, this holds as $\Omega(k)$ nodes can connect $\C$ with
  another component, and each of these nodes selects one of~$t$
  classes.
  
  For the case of less than $k/2$ short connector paths, it holds that
  there are $\Omega(k)$ long connector paths as in total $k$ connector
  paths exist according to \cref{lem:k-connector-paths}.  We use the
  remainder of this section to prove that a component $\C$ selects a
  long connector path with constant probability. The theorem follows
  with the result stated in \cref{prop:long-paths-constant-fraction}:
  At least one of the connector paths selects the required class with
  at least constant probability.
\end{proof}

\subsection{Enough short connector paths are sufficient}
\label{sec:enough-short-conn}

We state the lemma without a formal proof, which is given as part of
the proof of Lemma~4.4 in \cite{chgk-dcd-14}.

\begin{lemma}[part of Lemma~4.4 in \cite{chgk-dcd-14}]
  \label{lem:short-paths-constant-fraction}
  Given a class $i$ and a component $\C$ of layer $l>L \log n$ with at
  least $k/2$ short connector paths, $\C$ has at least one short
  connector path of class $I$ with probability at least $\delta$.
\end{lemma}

\subsection{Number of connector paths for type-2 node}
\label{sec:numb-conn-paths}

\begin{lemma}[Proposition~4.2 in \cite{chgk-anpvc-14}]
  \label{lem:type-2-node-on-leq-one-path}
  For an arbitrary class $i$ and a type-2 node $v$, $v$ lies on at
  most one long connector path of one component of $\C$.
\end{lemma}

Which implies the following Observation.

\begin{observation}
  \label{obs:type-2-node-leq-delta-components}
  Each type-2 node lies on at most $t \in \Theta(k)$ long connector paths.
\end{observation}

\subsection{Selected long connector paths are good}
\label{sec:lcp-good}

Let us sketch the proof of \cref{prop:long-paths-constant-fraction} in
the following.  The analysis is structured in three parts. The first
part considers how likely it is that a type-2 node $v$ selects another
class, given that the corresponding type-1 node is of the correct
class. In order to bound this probability, a random discard step for
long connector paths is introduced. This allows to show that all other
possible long connector paths through $v$ are discarded with constant
probability.  This results in an probability in the interval
$[1/4t,1/t]$ for the event that both internal vertices of a long
connector path of $\C$ select class $i$---independent of the class
selected by nodes on other long connector paths of $\C$.  In the
second part it is shown that the proven bounds hold even if the random
discard step is not used (this is required, as it is not used in the
algorithms).  The third and final step uses the independent bounds on
the probability of a long connector path to select class $i$ with a
tail bound to show that for at least one of the $\Omega(k)$ long
connector paths both internal vertices selected the same class $i$.

For more technical details we refer to
\cite{chgk-dcd-14}. All required adaptions are outlined Section~4.2 of
\cite{wolf-ba-14}.

\subsection{Identifying Connected Components}
\label{sec:ident-conn-comp}

To identifying and communicate through connected component, we use the
protocol described in \cite[Theorem B.2]{chgk-dcd-14}.  There are two
protocols that can be used, depending on the maximum diameter $D'$ of
the components in the virtual graph maximum diameter $D'$ of the
components, which is in $\O(\frac{n \log n}{k})$ whp. If it is
relatively small, i.e.
$\frac{n \log n}{k} = o(D+\sqrt{n \log n}\log^* n)$, a simple protocol
can be used, while a variation of a protocol to identify connector
components by Thurimella \cite{t-sldas-97} is used otherwise. 

Let us now consider the simpler variant. Each node transmits its
class, and the smalles node id it received so far (including its
own). Nodes discard received ids if they are transmitted by nodes
with different classes. After $D' = \O(\frac{n \log n}{k})$ rounds,
each node in each component received the smallest id of the component,
which is selected as the component id and the components root
node. The union of paths from the root to nodes of the components can
be used as communication tree in the component.

The more complex protocol, which is a variation of the algorithm to
identify connected components by Thurimella \cite{t-sldas-97} is
originally based on an minimum spanning tree (MST) algorithm by Garay,
Kutten and Peleg \cite{gkp-astda-98}, which was improved to the
current runtime bound by a new MST algorithm in
\cite{kp-fscsds-98}. The protocol allows each node in a network to
learn the smallest id in its component in $\O(D+\sqrt{n} \log^* n)$
rounds. The id of each virtual node $v_l$ (of layer $l$) is set to
(id$_v$, $l$, type), where id$_v$ is the id of the corresponding real
node, $l$ the virtual nodes layer, and type its type (either 1 or
2). The algorithm by Thurimella is executed on $\G$, which has a
diameter in $\O(D)$, and $\O(n \log n)$ nodes, resulting in $\O(D +
\sqrt{n\log n} \log^* n)$ rounds for identifying the connected components.


\end{appendix}


\end{document}